\documentclass[
10pt, 
a4paper, 
oneside, 
headinclude,footinclude, 
BCOR5mm, 
]{scrartcl}

\usepackage[
nochapters, 
beramono, 
eulermath,
pdfspacing, 
dottedtoc 
]{classicthesis} 

\usepackage{arsclassica} 

\usepackage[T1]{fontenc} 

\usepackage[utf8]{inputenc} 

\usepackage{graphicx} 
\graphicspath{{Figures/}} 
\usepackage{paralist} 
\usepackage{enumitem} 

\usepackage{lipsum} 

\usepackage{subfig} 

\usepackage{amsmath,amssymb,amsthm} 

\usepackage{varioref} 

\theoremstyle{definition} 

\theoremstyle{plain} 

\hypersetup{
colorlinks=true, breaklinks=true, bookmarks=true,bookmarksnumbered,
urlcolor=webbrown, linkcolor=RoyalBlue, citecolor=webgreen, 
pdftitle={}, 
pdfauthor={\textcopyright}, 
pdfsubject={}, 
pdfkeywords={}, 
pdfcreator={pdfLaTeX}, 
pdfproducer={LaTeX with hyperref and ClassicThesis} 
} 

\usepackage{lipsum} 
\usepackage{amsmath,amssymb,amsthm,amsfonts,scalerel}
\usepackage{color}
\usepackage{tikz}
\usepackage{pgfplots}
\usepackage{relsize}
\usepackage{float}
\usetikzlibrary{shapes,arrows}
\usepackage{bm,times}
\usepackage{verbatim}
\usetikzlibrary{calc,shadows,angles, arrows.meta, quotes}
\usepackage{enumitem}
\usepackage{eucal}         
\numberwithin{equation}{section}
\newtheorem{mydef}{Definition}[section]
\newtheorem{mytheo}{Theorem}[section]

\newtheorem{corollary}{Corollary}[section]
\newtheorem{myprop}{Proposition}[section]
\DeclareMathOperator{\tr}{tr}

\title{\normalfont\spacedallcaps{SPH Consistent With explicit LES}\\ 
	} 

\author{\spacedlowsmallcaps{Kalale Chola\textsuperscript{1} }} 

\date{\today} 

\begin{document}


\renewcommand{\sectionmark}[1]{\markright{\spacedlowsmallcaps{#1}}} 
\lehead{\mbox{\llap{\small\thepage\kern1em\color{halfgray} \vline}\color{halfgray}\hspace{0.5em}\rightmark\hfil}} 

\pagestyle{scrheadings} 


\maketitle 

\setcounter{tocdepth}{2} 

%
%


\section*{Abstract} 
\paragraph*{}The aim of this paper is to introduce a consistent velocity smoothing method for smoothed particle hydrodynamics (SPH). First the locally averaged Navier-Stokes equations are derived in a mathematically rigorous way to demonstrate the "missing" turbulent stress in standard SPH formulations. It is then shown that with a proper choice of velocity smoothing, SPH and large eddy simulation (LES) equivalently solve the same set of filtered equations that require closure approximations. The only difference between SPH and LES, as demonstrated in this paper is the representation of governing equations; for the former the equations are in integro-differential form whereas for the latter they are in differential form. One direct consequence of this equivalence between SPH and LES is that turbulence modeling techniques originally developed for LES can easily be adopted into this version of SPH. 

Our representation of the sub-grid stress tensor in integral form will provide insight into alternative approaches for dealing with turbulence modeling. Although the use of the Smagorinsky model is the common practice for turbulence modeling in LES, it will be shown that for SPH the most natural choice is to use approximate deconvolution methods. The other particularly fundamental consequence of our choice of velocity smoothing is that the resulting filtered equations are nonconservative in nature and a correct Lagrangian cannot be easily constructed.  





{\let\thefootnote\relax\footnotetext{\textsuperscript{1} \textit{Shintake Unit, OIST Graduate University, Okinawa, Japan}}}


\newpage 


\section{Introduction}
Smoothed particle hydrodynamics (SPH), pioneered by J.J. Monaghan, has become a popular particle method across many fields of science such as astrophysics, continuum fluid and solid mechanics, computer graphics e.t.c. The core elements of SPH can be found in Monaghan~\cite{Monaghan2005}, Price \cite{Price2012}. In this paper, the locally averaged continuum equations are first derived rigorously with emphasis on the consistent way of smoothing the velocity field and how to move the SPH particle.  

As observed by Monaghan \cite{Monaghan2002,Monaghan2009,Monaghan2011}, for high speed flows there are regions where SPH particles could stream through each other. In a continuum simulation this kind of behavior is unphysical as it implies multi-valued local velocity fields. A standard remedy for this is to artificially increase the viscosity which may lead to over-dissipation. The version of SPH to be presented in this paper 
is based on the principle that each fluid particle moves with a smoothed velocity so that counter-streaming is prevented. Furthermore, the filtering provides another advantage that local fluctuations are filtered out hence the local disorder is reduced. As it will be shown later on, the model does not conserve energy exactly. Another important feature is that the filtered conservation laws are consistent with those of explicit LES.

A feature of the proposed model is that the presence of the sub-particle stress (SPS) tensor should account for the effect of small scale motion on the smoothed flow fields. The SPH version of the SPS tensor is presented in integral form and when transformed into differential form it reduces to the sub-grid stress (SGS) tensor of LES. 

The discussion in this paper will proceed as follows. A filtering integral transform (FIT) is proposed. By applying the FIT to the continuum compressible Navier-Stokes equations (CNSEs) a set of smoothed or filtered equations consistent with those of explicit LES is derived. Finally, we discuss the procedure for turbulence modeling in a particle system by studying conditions under which the local fluctuations can be neglected. Numerical aspects of model will be presented elsewhere.

\section{Theoretical Development}
First, we consider the Lagrangian form of compressible Navier-Stokes equations (CNSEs) for a continuum.   
\begin{align}
\frac{d\rho}{dt} &= -\rho\nabla\cdot\mathbf{u} \label{eq2016:1}\\
\rho\frac{d\mathbf{u}}{dt}&= -\nabla p + \nabla\cdot\underline{\underline{\sigma}} + \rho\mathbf{b} \label{eq2016:2}\\ 
\kappa_{s}\frac{dp}{dt} &= -\nabla\cdot\mathbf{u} + \gamma\alpha\nabla\cdot(\kappa_{s}\nabla p) -\alpha\nabla\cdot\left(\frac{1}{\rho}\nabla\rho\right)+\frac{\beta}{\rho c_{p}}\Phi\label{eq2016:2aa}
\end{align} 
	  
where $\mathbf{u}$ is the fluid velocity, $\rho$ is fluid mass-density, $p$ is the pressure, $\underline{\underline{\sigma}}$ is viscous stress tensor, $\underline{\underline{S}}$ is the strain rate tensor,  $\Phi =\underline{\underline{\sigma}}:\underline{\underline{S}}$ is the viscous dissipation, $\alpha$ is the thermal diffusivity, $\beta$ is the volumetric thermal expansivity, $c_{p}$ is the specific heat capacity at constant pressure, $c_{v}$ is the specific heat capacity at constant volume, $\gamma=c_p/c_v=K_{S}/K_{T}$ is the adiabatic index. Here, $K_{S}$ is the adiabatic incompressibility modulus and $K_{T}$ is the isothermal incompressible modulus. For water and air, $\gamma$ is 7.0 and 1.4 respectively. We further assume that the adiabatic incompressibility modulus varies linearly with pressure so that 
\begin{align}
K_{S}= K_{S,0} + p\equiv\frac{1}{\kappa_{s}}
\end{align}
For an ideal gas, it then follows that $K_{S}$ is given by,
\begin{align}
K_{S}= \gamma K_{T}=-\gamma v\left(\frac{\partial p}{\partial v}\right)_{T}=\gamma p
\end{align}
meaning that for an ideal gas, $K_{S,0}=0$. For gases and liquids in general, we have that the incompressibility modulus under standard conditions $K_{S,0}$ is related to the standard speed of sound $c_0$ as $K_{S,0} = \rho_{0}c_{0}^{2}$.

For a Newtonian fluid, the Cauchy stress tensor $\underline{\underline{\tau}}$ can be expressed in terms of a deviatoric stress and a normal stress as
\begin{eqnarray}
\underline{\underline{\tau}}&= -p\underline{\underline{1}}+ \underline{\underline{\sigma}}=-p\underline{\underline{1}}+\nu\rho (\nabla\mathbf{u}+\nabla\mathbf{u}^{\textup{T}}-\frac{2}{d}\nabla\cdot\mathbf{u}\underline{\underline{1}})\label{eqn:201619}
\end{eqnarray}
where $\nu$ is the kinematic viscosity and $d$ is the space dimension. Note that the Stokes hypothesis for a zero bulk viscosity has been assumed for simplicity.

These PDEs are valid over a continuum where the fluid variables are assumed to be smooth and continuous.

\subsection{Filtering problem}  
Given the continuum or disordered field\{$\rho(\mathbf{r})$, $p(\mathbf{r})$, $\mathbf{u}(\mathbf{r})$\} defined on a domain $\Omega$, compute local average fields \{$\langle\rho_{h}(\mathbf{r})\rangle$, $\langle p_{h}(\mathbf{r})\rangle$, $\widetilde{\mathbf{u}}_{h}(\mathbf{r})$\} that faithfully represent the behavior of the disordered field on scales above some, user defined, filter length (here denoted $h$) and which truncates scales smaller than $\mathcal{O}(h)$.

The main goal is to transform the governing equations (\ref{eq2016:1}), (\ref{eq2016:2}) and (\ref{eq2016:2aa}) into integro-differential equations. After this transformation, we can then invoke a particle discretization by replacing integrals with summations; we then obtain our version of SPH which are a set of ordinary differential equations in time. On the other hand, if the integro-differential are now transformed into partial differential equations, we obtain the governing equations for LES thereby establishing the equivalence between the SPH and LES approaches.
  
\subsection{Local Averaging}
In order to derive the conservation laws for a discrete fluid, we need locally averaged variables. Using ideas of distribution theory, physical attributes of the material elements of the discrete fluid such as mass density, linear momentum density, and velocity are replaced by local mean variables obtained by averaging the point variables over small local regions (test spaces) containing many material elements but are still small compared with the scale of macroscopic variation from point to point within the system, Jackson \cite{Jackson1967}.

In science and engineering, in order to measure the amount of a physical quantity such as temperature at a single point, one needs a probe which can extract data only from that single point of interest in space. Since this is not possible practically, there is no way we can correlate the experimental data with theoretical predictions. Therefore, a true macroscopic quantity is by necessity an average over some spatial region surrounding the continuum point where it is nominally defined, Admal et al.~\cite{Admal2010}. 

Consider a function $T(\mathbf{r})$ as representing a value of the physical variable at a particular point $\mathbf{r}$ in space. Is this a realistic thing to do? What can we measure?

Suppose $T(\mathbf{r})$ represents temperature at a point $\mathbf{r}$ in a room $\Omega$. The temperature can be measured with a thermometer by placing the bulb at the point $\mathbf{r}$. Unlike the point, the bulb has nonzero size, so what the thermometer actually measures is the mean temperature over a small region of space $\Omega(\mathbf{r})\in\Omega\subset\mathbb{R}^{n}$. So really, the thermometer measures
\begin{align}
w &\mapsto\int_{\Omega(\mathbf{r})}T(\mathbf{r})w(\mathbf{r})d^{3}\mathbf{r}
\end{align}  
where $w(\mathbf{r})$ is a test function representing the properties of the probe and its length scale. Physically~$w(\mathbf{r})$ depends on the nature of the thermometer or probe and where you place it. The test function~$w(\mathbf{r})$ will tend to be "concentrated" near the location of the thermometer bulb or probe and nearly zero once you are sufficiently far away from the bulb. To say this is an average requires that  $^\forall\mathbf{r}\in\Omega(\mathbf{r})\subset\Omega$;

\begin{enumerate}[noitemsep]
\item[(i)]semi-positive definiteness
\begin{align}
w(\mathbf{r}) &\ge 0
\end{align}
So that a positive measure e.g mass density will have a local average that is likewise a positive measure.
\item[(ii)] completeness
\begin{align}
\int_{\Omega(\mathbf{r})}w(\mathbf{r})d^{n}\mathbf{r} &=1
\end{align}
This condition ensures that the correct macroscopic or continuum quantity is obtained when the corresponding locally averaged quantity is uniform. 
\end{enumerate}
So it would be more meaningful to discuss things like the value of T about a point, ~$ w\mapsto T(w)$~than things like the value of $T$ at a point, $\mathbf{r}\mapsto T(\mathbf{r})$. 

\begin{mydef}[fluid domain]
Let $\Omega\subseteq\mathbb{R}^{n}$ be an open set bounded by a smooth surface $\partial\Omega$. If $\Omega$ is supposed to be "filled with a fluid" such that the mass-density
\begin{align}
\rho(\mathbf{r}) > 0, \qquad ^\forall \mathbf{r}\in\Omega
\end{align}
then we call $\Omega$ a fluid domain.
\end{mydef}

\begin{mydef}[smoothness and compactness]
A real-valued function is said to be \emph{smooth} if it is \emph{infinitely differentiable}. Let $w$ be a smooth and compact function over the test space. We then write $w\in C^{\infty}_{c}(\Omega(\mathbf{r}))$. 
\end{mydef}

\begin{mydef}[Test Space] 
Let $\Omega\subset\mathbb{R}^{n}$ be a given body. For a target particle located at $\mathbf{r}$, we denote its test space as $\Omega_{h}(\mathbf{r})$ bounded by a test surface $\partial\Omega(\mathbf{r})$. The test space is thus the domain of influence of the target particle.
\begin{align}
\Omega_{h}(\mathbf{r}) &= \left\{ \mathbf{r},\mathbf{r}^{\prime}\in\mathbb{R}^{d}\big\vert~w_{h}(\mathbf{r}-\mathbf{r}^{\prime})\geq 0,~\vert\vert\mathbf{r}-\mathbf{r}^{\prime}\vert\vert\leq \kappa h \ , \kappa\in\mathbb{R}^{+}\right\}
\end{align}
\end{mydef}
\subsection{Invertibility of operators}
A linear operator $\hat{C}_{h}$ has a left inverse $\hat{L}_{h}$ if $\hat{L}_{h}\hat{C}_{h}=\hat{1}$ and a right inverse $\hat{R}_{h}$ if $\hat{C}_{h}\hat{R}_{h}=\hat{1}$. Furthermore, if both exist, we have that
\begin{eqnarray}
\hat{L}_{h}&=&\hat{L}_{h}\hat{1}
=\left(\hat{L}_{h}\hat{C}_{h}\right)\hat{R}_{h}=\hat{1}\hat{R}_{h}=\hat{R}_{h}
\end{eqnarray}
which is very useful for infinite dimensional vector spaces. Therefore, if an operator has both a left and a right inverse, then the left inverse is the same as the right inverse and we say that the operator is invertible.

A left inverse exists if the action of the operator on some input vector does not result in irreparable damage so that whatever remains still contains enough information that some linear operator $\hat{L}_{h}$ can restore our original input vector and give back the identity operator. This condition of irreparable damage i.e. not losing information is asking whether the operator $\hat{C}_{h}$ is injective. There exists a left inverse if and only if $\hat{C}_{h}$ is injective. 

For the right inverse to exist, the situation is dual to that of the left inverse; a right inverse exists if $\hat{C}_{h}$ is surjective. 

In particular, if $\hat{C}_{h}$ is an operator in a finite dimensional space $V$ i.e. $dim V < \infty$ then the following is true.
\begin{align}
\hat{C}_{h}\text{is injective}\Longleftrightarrow\hat{C}_{h}\text{is surjective}\Longleftrightarrow\hat{C}_{h}\text{is invertible}
\end{align}
$\hat{C}_{h}$ is surjective if and only if $\hat{C}_{h}$ is injective because failure to be injective and failure to be surjective are both equivalent to loss of information. This can be seen from the dimension formula for a finite vector space
\begin{align}
\text{dim}V &= \text{dim}\left(\text{null} \hat{C}_{h}\right) + \text{dim}\left(\text{range} \hat{C}_{h}\right)\label{deq:1c}
\end{align}
It is clear that if dim $\left(\text{null} \hat{C}_{h}\right)= 0$, then $dim \left(\text{range}  \hat{C}_{h}\right)$ is the whole vector space. Furthermore, if dim $\left(\text{null} \hat{C}_{h}\right)\neq 0$, then $\text{dim} \left(\text{range} \hat{C}_{h}\right)$ is not the whole vector space.

The equivalence (\ref{deq:1c}) breaks down if the vector space is infinite dimensional. For infinite dimensional vector spaces injectivity and surjectivity are not equivalent since each can fail independently. 
\begin{align}
\hat{C}_{h}\text{is invertible}\Longleftrightarrow\hat{C}_{h}\text{is injective and surjective}
\end{align}
In the discussions that follow, we are going to assume that the convolution operator $\hat{C}_{h}$ is invertible and arbitrary. 

\subsection{Resolution of Identity}
Position states for describing the continuum mechanics of material elements or fluid particles moving in an n-dimensional space $\mathbf{r},\mathbf{r}^{\prime}\in\Omega_{h}(\mathbf{r})\subset\mathbb{R}^{n}$ are defined as follows
\begin{align}
\text{position basis state :}\quad \vert\mathbf{r}\rangle\equiv\vert x,y,z\rangle, \quad ^\forall x,y,z\in\mathbb{R}
\end{align} 
where the label $\mathbf{r}=\mathbf{r}(t)$ in the ket is the position of a material element or its trajectory. Since particle trajectories are continuous and we position states $\vert\mathbf{r}\rangle$ for all $\mathbf{r}$ to form a basis, we are dealing with a non-denumerable or infinite basis. Therefore the ket is a vector in infinite dimensional vector space of states of the theory. The $\vert$ $\rangle$ enclosing the label of the position eigenstates plays a crucial role: it helps us to see that that object lives in an infinite dimensional vector space. Basis states with different values of $\mathbf{r}$ are different vectors in the state space.
The inner product must be defined, so we take
\begin{align}
\langle\mathbf{r}\vert\mathbf{r}^{\prime}\rangle &=\delta(\mathbf{r}-\mathbf{r}^{\prime})\equiv\delta(x-x^{\prime})\delta(y-y^{\prime})\delta(z-z^{\prime}), \quad ^\forall x,y,z\in\mathbb{R}\label{deq:1e}
\end{align}
It then follows that position states with different positions are orthogonal to each other. The norm of position states is infinite: $\langle\mathbf{r}\vert\mathbf{r}\rangle=\delta(0)=\infty$, so these are not allowed states of particles. This also implies that no two fluid particles can occupy the same position at the same time.

We visualize the state $\vert\mathbf{r}\rangle$ as the state of a fluid particle perfectly localized at position $\mathbf{r}$, but this is just an idealization.

Normalizable states can be easily constructed by using superposition of position states using the completeness relation or resolution of identity
\begin{align}
\int_{\Omega_{h}(\mathbf{r})}\vert\mathbf{r}^{\prime}\rangle\langle\mathbf{r}^{\prime}\vert d\Omega(\mathbf{r}^{\prime})&=\hat{1}\label{deq:b1}
\end{align}
which is consistent with the inner product (\ref{deq:1e}). 
At this point we introduce a new state vector describing the mass density of a fluid particle  
\begin{align}
\text{mass density state:}\quad \vert\rho\rangle, \quad \rho > 0~\text{on}~\Omega_{h}(\mathbf{r})
\end{align}
To project operators into function space we simply take an overlap of the coordinate basis state $\langle\mathbf{r}\vert$ with the state $\vert\rho\rangle$ yielding the value of the real-valued function $\rho$ at position $\mathbf{r}$; namely
\begin{align}
\rho(\mathbf{r})&=\langle\mathbf{r}\vert\rho\rangle\in\mathbb{R}
\end{align} 

\section{Filtering Integral Transform (FIT)}
The convolution or filtering problem can be stated formally as follows:
Given the continuum field\{$\rho(\mathbf{r})$, $p(\mathbf{r})$, $\mathbf{u}(\mathbf{r})$\} defined on a domain $\Omega$, compute local approximations \{$\langle\rho_{h}(\mathbf{r})\rangle$, $\langle p_{h}(\mathbf{r})\rangle$, $\widetilde{\mathbf{u}}_{h}(\mathbf{r})$\} which faithfully represent the behavior of the continuum field on scales above some, user defined, filter length (here denoted $h$) and which truncates scales smaller than $\mathcal{O}(h)$.

Standard SPH is based upon the fundamental principle that any field $\Phi:\Omega_{h}(\mathbf{r})\to\mathbb{R}$ can be expressed by an integral interpolant $\langle\Phi_{h}(\mathbf{r})\rangle:=\int_{\Omega_{h}(\mathbf{r})}\Phi(\mathbf{r}^{\prime})w_{h}(\mathbf{r}-\mathbf{r}^{\prime})d^{\nu}\mathbf{r}^{\prime}$. However, following the discussion from the previous section, if the particles are uniformly distributed locally and the kernel support does not intersect the domain boundaries, the local approximation becomes $\langle\Phi_{h}(\mathbf{r})\rangle=\Phi(\mathbf{r})+\mathcal{O}(h^{2})$ meaning SPH is second order accurate. The integral interpolant then assumes the standard form and constitutes the most fundamental principle as  $\Phi(\mathbf{r})=\int_{\Omega_{h}(\mathbf{r})}\Phi(\mathbf{r}^{\prime})w_{h}(\mathbf{r}-\mathbf{r}^{\prime})d^{\nu}\mathbf{r}^{\prime}+\mathcal{O}(h^{2})$. Formally, this is called zeroth order deconvolution which results in second order accuracy when bell-shaped convolution filters are used. Unfortunately, the local particle distribution is no longer uniform, particles get disordered during the evolution of a system. Therefore, error estimation becomes a very challenging task. 

In this section, we demonstrate that application of the FIT to the CNSEs leads to a version SPH that is consistent with explicit LES. The concept of local fluctuations and the relation with local uniformity of particle distribution is discussed. The section concludes by stressing the implications of the choice of velocity smoothing on the complexity of the mathematical structure of the smoothed CNSEs.

Assuming that the convolution operator $\hat{C}_{h}:L^{2}(\Omega_{h})\to L^{2}(\Omega_{h})$ is invertible. Then, in operator form, the action of this operator on an input state $\vert\rho\rangle$ yields a new state given by
\begin{align}
\vert\bar{\rho}_{h}\rangle &=\hat{C}_{h}\vert\rho\rangle\label{deq:bb5}
\end{align}
\begin{figure}[H]
	\centering
	\tikzstyle{int}=[draw, fill=blue!20, minimum size=8em]
	\tikzstyle{init} = [pin edge={to-,thin,black}]
	\begin{tikzpicture}[node distance=5.5cm,auto,>=latex']
	\node [int] (a) {$\hat{C}_{h}$};
	\node (b) [left of=a,node distance=3cm, coordinate] {a};
	\node [coordinate] (end) [right of=a, node distance=3cm]{};
	\path[->] (b) edge node {$\vert\rho\rangle$} (a);
	\draw[->] (a) edge node {$\vert\bar{\rho}_{h}\rangle$} (end);
	\end{tikzpicture}
	
	\caption[Filtering Integral Transform]{Convolution process. The output state $\vert\bar{\rho}_{h}\rangle$ acquires the smoothness property of the convolution operator as well as scale-dependence.}
	\label{fig:convolution}
\end{figure}
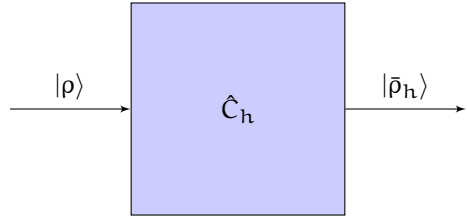  
The above result can be pictured in function space by multiplying from the left by the bra vector $\langle\mathbf{r}\vert$ and inserting a complete set of states as follows
\begin{align}
\langle\mathbf{r}\vert\bar{\rho}_{h}\rangle
&=\langle\mathbf{r}\vert\hat{C}_{h}\hat{1}\vert\rho\rangle\nonumber\\ 
&=\langle\mathbf{r}\vert\hat{C}_{h}\int_{\Omega_{h}(\mathbf{r})}\vert \mathbf{r}^{\prime}\rangle\langle \mathbf{r}^{\prime}\vert d\Omega(\mathbf{r}^{\prime})\vert\rho\rangle\nonumber\\ 
&= \int_{\Omega_{h}(\mathbf{r})}\langle \mathbf{r}\vert\hat{C}_{h}\vert \mathbf{r}^{\prime}\rangle\langle \mathbf{r}^{\prime}\vert\rho\rangle d\Omega(\mathbf{r}^{\prime})\nonumber\\
\langle{\rho}_{h}(\mathbf{r})\rangle &= \int_{\Omega_{h}(\mathbf{r})}\rho (\mathbf{r}^{\prime})w_{h}(\mathbf{r}-\mathbf{r}^{\prime}) d\Omega(\mathbf{r}^{\prime})\label{deq:bb4}
\end{align}
where $\langle r\vert\bar{\rho}_{h}\rangle = \langle{\rho}_{h}(\mathbf{r})\rangle$ is the the smoothed or filtered mass density. The convolution filter is defined as the "matrix element" of the convolution operator $\hat{C}_{h}$, i.e. $w_{h}(\mathbf{r}-\mathbf{r}^{\prime}):=\langle \mathbf{r}\vert\hat{C}_{h}\vert\mathbf{r}^{\prime}\rangle$.

Similar results can be obtained for the momentum density and pressure. With brevity, we present these in the following proposition.

\begin{myprop}[FIT for fluids]\label{prop:conv}
	Let $\Omega_{h}(\mathbf{r})$ be a locally compact space within the fluid domain $\Omega$. Then the filtered mass density, momentum density and pressure are given by the FIT; for each $w_{h}\in C^{\infty}_{c}(\Omega_{h})$
	\begin{align}
	\langle\rho_{h}(\mathbf{r})\rangle &=\int_{\Omega_{h}(\mathbf{r})}\rho(\mathbf{r}^{\prime})w_{h}(\mathbf{r}-\mathbf{r}^{\prime})d\Omega(\mathbf{r}^{\prime})\label{deq:1}\\
	\langle\rho_{h}(\mathbf{r})\rangle\widetilde{\mathbf{u}}_{h}(\mathbf{r}) &=\int_{\Omega_{h}(\mathbf{r})}\rho(\mathbf{r}^{\prime})\mathbf{u}(\mathbf{r}^{\prime})w_{h}(\mathbf{r}-\mathbf{r}^{\prime})d\Omega(\mathbf{r}^{\prime})\label{deq:1p}\\
	\langle p_{h}(\mathbf{r})\rangle &=\int_{\Omega_{h}(\mathbf{r})}p(\mathbf{r}^{\prime})w_{h}(\mathbf{r}-\mathbf{r}^{\prime})d\Omega(\mathbf{r}^{\prime})\label{deq:p1}
	\end{align} 
\end{myprop}
The smoothed field \{$\langle\rho_{h}(\mathbf{r})\rangle$, $\langle p_{h}(\mathbf{r})\rangle$, $\widetilde{\mathbf{u}}_{h}(\mathbf{r})$\} represent the interaction of fluid particles located at $\mathbf{r}$, $\mathbf{r}^{\prime}\in\Omega_{h}(\mathbf{r})$. Furthermore, the choice of the velocity smoothing here arises from the physical consideration that the smoothed velocity $\widetilde{\mathbf{u}}_{h}:=\langle\mathbf{P}_{h}\rangle/\langle\rho_{h}\rangle$ where $\mathbf{P}$ is the momentum density.

\subsection{Effect of the convolution operator: geometrical analysis}\label{subsub:geometrical}
In the SPH method it is well known that the most important element of the method is the convolution kernel. Here, a criteria for measuring the effect of this convolution kernel through a geometrical analysis is presented.  
\begin{figure}[H]
	\centering
	\noindent\begin{tikzpicture}[> = stealth]
	\draw[ultra thick,red, ->]  (1,1) -- node[below] {$\vert\rho\rangle$} +(4,1);
	\draw[ultra thick,red, ->]  (1,1) -- node[above,sloped] {$\vert\bar{\rho}_{h}\rangle=\hat{C}_{h}\vert\rho\rangle$} +(6,4);
	\draw[densely dotted,blue,-] (5,2) -- node[below] {$U_{\vert\rho\rangle}$} +(4,1);
	\draw[ultra thick,blue, ->]  (7.5882,2.6471) -- node[right] {$\vert\rho_{h}^{\perp}\rangle$} +(-0.5882,2.3529);
	\end{tikzpicture}
	\caption[Pictorial view of convolution operator effect]{The mass density state $\vert\rho\rangle$~and the subspace $U_{\vert\rho\rangle}$ generated by it. The effect of the convolution operator is to shift the mass density state $\vert\rho\rangle$ into a new vector subspace spanned by the deformed state $\vert\bar{\rho}_{h}\rangle$. The norm of the orthogonal complement $\vert\rho^{\perp}\rangle$ is the error in resolving identity using a convolution operator.}
	\label{fig:convVector}
\end{figure}
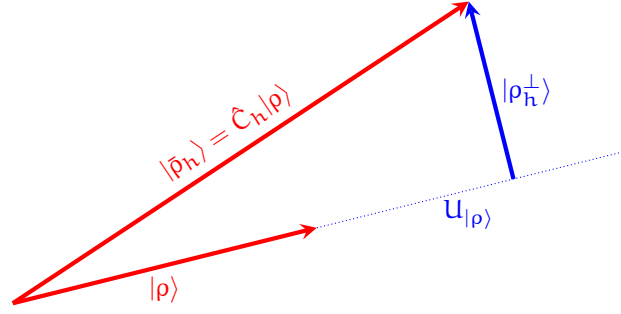
Figure~\ref{fig:convVector} demonstrates the effect of the convolution operator. Consider the vector subspace $U_{\vert\rho\rangle}$ generated or spanned by the state $\vert\rho\rangle$. When the convolution operator $\hat{C}_{h}$ acts on the undeformed state $\vert\rho\rangle$ the output is a deformed state~$\vert\bar{\rho}_{h}\rangle$ that is shifted in a different direction due to the fact that $\vert\rho\rangle$ is not an eigenstate of the convolution operator. If we define an orthogonal projector to the vector subspace $U_{\vert\rho\rangle}$ 
\begin{align}
\hat{P}_{U_{\vert\rho\rangle}} :=\frac{\vert\rho\rangle\langle\rho\vert}{\langle\rho\vert\rho\rangle}
\end{align}
Then an orthogonal projection of the state $\vert\bar{\rho}_{h}\rangle$ onto $U_{\vert\rho\rangle}$ yields 
\begin{align}
\hat{P}_{U_{\vert\rho\rangle}}\vert\bar{\rho}_{h}\rangle &=\frac{\langle\rho\vert\bar{\rho}_{h}\rangle}{\langle\rho\vert\rho\rangle}\vert\rho\rangle
\end{align}
which is the component of the deformed state $\vert\bar{\rho}_{h}\rangle$ along the undeformed state~$\vert\rho\rangle$.
The orthogonal complement state $\vert\rho^{\perp}_{h}\rangle$ spans the subspace $U^{\perp}_{\vert\rho\rangle}$ and is given by
\begin{align}
\vert\rho^{\perp}_{h}\rangle &=\vert\bar{\rho}_{h}\rangle-\frac{\langle\rho\vert\bar{\rho}_{h}\rangle}{\langle\rho\vert\rho\rangle}\vert\rho\rangle
\end{align}
We then claim that the error or uncertainty in approximating identity using the convolution operator is given by the norm of the perpendicular state.
\begin{align}
\Delta\hat{C}_{h}(\rho)^{2}&=\langle\rho^{\perp}_{h}\vert\rho^{\perp}_{h}\rangle =\langle\bar{\rho}_{h}\vert\bar{\rho}_{h}\rangle-\frac{\vert\langle\rho\vert\bar{\rho}_{h}\rangle\vert^{2}}{\langle\rho\vert\rho\rangle}\label{deq:2c}
\end{align}
Thus a good convolution operator (kernel) is one that minimizes the error $\Delta\hat{C}_{h}(\rho)$. It is important to also note that in the continuum limit this error is zero as the convolution operator becomes identical to the identity operator. However, for any practical convolution operator (kernel) this error is non-zero. Therefore this gives us a measure of how well a convolution operator approximates identity, or equivalently, how well a convolution kernel approximates the Dirac-delta function. 

In function space (\ref{deq:2c}) is expressed as    
\begin{align}
\Delta\hat{C}_{h}(\rho)^{2}&=\int_{\Omega}\vert\langle\rho_{h}(\mathbf{r}^{\prime})\rangle\vert^{2}d\Omega(\mathbf{r}^{\prime})-\frac{\vert\int_{\Omega}\rho(\mathbf{r}^{\prime})\langle\rho_{h}(\mathbf{r}^{\prime})\rangle d\Omega(\mathbf{r}^{\prime})\vert^{2}}{\int_{\Omega}\vert\rho(\mathbf{r}^{\prime})\vert^{2}d\Omega(\mathbf{r}^{\prime})}\label{deq:3c}
\end{align}
For spherically symmetric convolution kernels, a Taylor expansion of $\rho(\mathbf{r}^{\prime})$ about $\mathbf{r}^{\prime}$ shows that the error is of second order $\Delta\hat{C}_{h}(\rho)\equiv\mathcal{O}(h^{2})$.
Clearly, minimizing this error is thus fundamental for better approximation of the resolution of identity using a convolution operator (kernel).  

\section{Application of FIT}
Our choice of the  filtering procedure defined by proposition \ref{prop:conv} is to derive a set of integro-differential equations that are in principle "equivalent" to the filtered equations of explicit LES. Therefore, well established techniques such as turbulence of sub-grid stress tensors  in the LES context can be applied to this version of SPH. We first prepare a useful theorem that is extensively used in the this section.
\begin{mytheo}[Reynolds' transport theorem]\label{theo:rtt}
	Let f:$\bar{\mathcal{D}}\times\mathcal{I}\to\mathbb{R}$ be a smooth and continuous function and $\Omega_{t}:=\varphi(\Omega_{0},t)\subseteq\mathcal{D}$, then for each $t\in\mathcal{I}$ and $\Omega_{0}\subseteq\mathcal{D}$ as an arbitrary reference fluid domain. Then,
	\begin{align}
	\frac{d}{dt}\int_{\Omega_{t}}f(\mathbf{r},t)d^{n}\mathbf{r}&=\int_{\Omega_{t}}\left(\frac{d}{dt}f(\mathbf{r},t)+f(\mathbf{r},t)\nabla\cdot\mathbf{u}(\mathbf{r},t)\right)d^{n}\mathbf{r}
	\end{align}
\end{mytheo} 

\subsection{Filtered continuity equation}
To obtain a consistent set of filtered CNSEs, we proceed as follows. Consider a test particle located at position $\mathbf{r}$ with a test space (domain of influence) $\Omega_{h}(\mathbf{r})$, a compact space. When a measurement is performed over this test space to determine the test particle's physical attributes, the measurement outcomes are the local field approximations $\{\langle\rho_{h}(\mathbf{r})\rangle, \langle p_{h}(\mathbf{r})\rangle, \widetilde{\mathbf{u}}_{h}(\mathbf{r})\}$ defined by the FIT of proposition \ref{prop:conv}. These are the observables or macroscopic variables that would faithfully represent the behavior of the underlying, disordered continuum fields $\{\rho(\mathbf{r}), p(\mathbf{r}), \mathbf{u}(\mathbf{r})\}$ above the filter width $h$ consequently truncating scales smaller than $\mathcal{O}(h)$. 
Let the set $\{\mathcal{P}_{j}\vert~\vert\vert\mathbf{r}-\mathbf{r}^{\prime}<2h\vert\vert,~j = 1,...,N_{n}\}$ represent all other material elements within the test space. For this set of support material particles, the point-form of physical attributes $\{\rho(\mathbf{r}_{j}), p(\mathbf{r}_{j}), \mathbf{u}(\mathbf{r}_{j}) \vert~ j=1,...,N_{n}\}\in\textup{L}^{2}(\Omega_{h};\mathbb{R}^{n})$ are given. We then have after testing the continuity equation 

\begin{align}
& \int_{\Omega_{h}(\mathbf{r})}\left(\frac{d\rho(\mathbf{r}_{j})}{dt}+\rho(\mathbf{r}_{j})\nabla_{j}\cdot\mathbf{u}(\mathbf{r}_{j})\right)w_{h}(\mathbf{r}-\mathbf{r}_{j})d\Omega(\mathbf{r}_{j})\nonumber\\
& \qquad {} \stackrel{N_n\uparrow\infty}{\longrightarrow} \int_{\Omega_{h}(\mathbf{r})}\left(\frac{d\rho(\mathbf{r}^{\prime})}{dt}+\rho(\mathbf{r}^{\prime})\nabla^{\prime}\cdot\mathbf{u}(\mathbf{r}^{\prime})\right)w_{h}(\mathbf{r}-\mathbf{r}^{\prime})d\Omega(\mathbf{r}^{\prime})
\end{align}
where now $\mathbf{r}^{\prime}\in\Omega(\mathbf{r})$ is the position of a fluid particle in the test space. Here \emph{weak convergence} has been assumed, provided the test space is sufficiently populated uniformly by fluid particles. Then the weak form of the continuity equation becomes
\begin{align}
\int_{\Omega_{h}(\mathbf{r})}\left(\frac{d\rho(\mathbf{r}^{\prime})}{dt}+\rho(\mathbf{r}^{\prime})\nabla^{\prime}\cdot\mathbf{u}(\mathbf{r}^{\prime})\right)w_{h}(\mathbf{r}-\mathbf{r}^{\prime})d^{d}(\mathbf{r}^{\prime})&=0\quad{} ^\forall w_{h}\in C_{c}^{\infty}(\Omega_{h})\label{eqn:2017c1}
\end{align}
which is further simplified to 
\begin{align}
\int_{\Omega_{h}(\mathbf{r})}\left\{\frac{d}{dt}\bigg(\rho(\mathbf{r}^{\prime})w_{h}\bigg)+\bigg(\rho(\mathbf{r}^{\prime})w_{h}\bigg)\nabla^{\prime}\cdot\mathbf{u}(\mathbf{r}^{\prime})\right\}d^{d}\mathbf{r}^{\prime}&=\int_{\Omega_{h}(\mathbf{r})}\rho(\mathbf{r}^{\prime})\frac{dw_{h}}{dt}d^{d}\mathbf{r}^{\prime}\label{eqn:2017c6}
\end{align}
Using the Reynolds transport theorem \ref{theo:rtt} it is possible to simplify (\ref{eqn:2017c6}) even further yielding 
\begin{align}
\frac{d}{dt}\int_{\Omega_{h}(\mathbf{r})}\rho(\mathbf{r}^{\prime})w_{h}d\Omega(\mathbf{r}^{\prime})&=
\int_{\Omega_{h}(\mathbf{r})}\rho(\mathbf{r}^{\prime})\frac{dw_{h}}{dt}d\Omega(\mathbf{r}^{\prime})\label{eqn:2017c7}
\end{align}
Since the goal is to compute the local approximations for a target particle using the disordered fields $\{\rho(\mathbf{r}), p(\mathbf{r}), \mathbf{u}(\mathbf{r})\}$, we move the target particle with the local velocity. Accordingly,
\begin{align}
\textup{target particle:}\quad \frac{d\mathbf{r}}{dt}:=\widetilde{\mathbf{u}}_{h}(\mathbf{r})\quad \textup{support particle:}\quad\frac{d\mathbf{r}^{\prime}}{dt}:=\mathbf{u}(\mathbf{r}^{\prime})
\end{align}
The right hand side of the above equation is finally simplified with the help of the chain rule of calculus.
\begin{align}
\frac{d}{dt}\int_{\Omega_{h}(\mathbf{r})}\rho(\mathbf{r}^{\prime})w_{h}d\Omega(\mathbf{r}^{\prime})&=
\int_{\Omega_{h}(\mathbf{r})}\rho(\mathbf{r}^{\prime})\frac{dw_{h}}{dt}d\Omega(\mathbf{r}^{\prime})\nonumber\\
\frac{d}{dt}\int_{\Omega(\mathbf{r})}\rho(\mathbf{r}^{\prime})w_{h}d\Omega(\mathbf{r}^{\prime})&=
\int_{\Omega(\mathbf{r})}\rho(\mathbf{r}^{\prime})\left(\frac{d\mathbf{r}}{dt}\cdot\frac{\partial w_{h}}{\partial\mathbf{r}}+\frac{d\mathbf{r}^{\prime}}{dt}\cdot\frac{\partial w_{h}}{\partial\mathbf{r}^{\prime}}\right)d\Omega(\mathbf{r}^{\prime})\nonumber\\
\therefore\frac{d}{dt}\langle\rho_{h}(\mathbf{r}\rangle &=
\int_{\Omega(\mathbf{r})}\rho(\mathbf{r}^{\prime})\bigg(\widetilde{\mathbf{u}}_{h}(\mathbf{r})-\mathbf{u}(\mathbf{r}^{\prime})\bigg)\cdot\nabla w_{h}d\Omega(\mathbf{r}^{\prime})\label{deq:2018d4}
\end{align}

Equation (\ref{deq:2018d4}) is the filtered form of the continuity equation in integro-differential form. Using the FIT, it is easy to prove that, when expressed in differential form, this equation reduces to the canonical form of the continuity equation (\ref{eq2016:1}) but with the unfiltered variables replaced by the filtered ones. To prove this, we unplug the space derivative $\nabla$ from under the integral in (\ref{deq:2018d4}) and using the FIT of proposition \ref{prop:conv} as follows;
\begin{proof}
	Unplugging the space derivatives from the integral in (\ref{deq:2018d4}) we have,
	\begin{align}
	\frac{d}{dt}\langle\rho_{h}(\mathbf{r})\rangle&=\widetilde{\mathbf{u}}_{h}(\mathbf{r})\cdot\nabla\int_{\Omega_{h}(\mathbf{r})}\rho(\mathbf{r}^{\prime}) w_{h}(\mathbf{r}-\mathbf{r}^{\prime}) d\Omega(\mathbf{r}^{\prime}) \nonumber \\ & \quad -\nabla\cdot\int_{\Omega_{h}(\mathbf{r})}\rho(\mathbf{r}^{\prime})\mathbf{u}(\mathbf{r}^{\prime})w_{h}(\mathbf{r}-\mathbf{r}^{\prime})d\Omega(\mathbf{r}^{\prime}) \nonumber \\
	&=\widetilde{\mathbf{u}}_{h}(\mathbf{r})\cdot\nabla\langle\rho_{h}(\mathbf{r})\rangle \nonumber-\nabla\cdot\big(\langle\rho_{h}(\mathbf{r})\rangle\widetilde{\mathbf{u}}_{h}(\mathbf{r})\big)\quad\text{by the FIT}\nonumber \\  \therefore\frac{d}{dt}\langle\rho_{h}(\mathbf{r})\rangle{}&=-\langle\rho_{h}(\mathbf{r})\rangle\nabla\cdot\widetilde{\mathbf{u}}_{h}(\mathbf{r})\label{deq:2018d5}
	\end{align}
	and we arrive at the filtered form of the point form of the continuum continuity equation, consistent with explicit LES.
\end{proof}

Furthermore, after the filtering process, the material derivative becomes
\begin{align}
\frac{d}{dt} &=\frac{\partial}{\partial t} +\widetilde{\mathbf{u}}_{h}\cdot\nabla
\end{align}
\begin{corollary}[filtered velocity divergence]
	Due to the equivalence  of (\ref{deq:2018d4}) and (\ref{deq:2018d5}), the velocity divergence in a continuum can be calculated as an integral
	\begin{align}
	-\langle\rho_{h}(\mathbf{r})\rangle\nabla\cdot\widetilde{\mathbf{u}}_{h}(\mathbf{r})&=\int_{\Omega(\mathbf{r})}\rho(\mathbf{r}^{\prime})\bigg(\widetilde{\mathbf{u}}_{h}(\mathbf{r})-\mathbf{u}(\mathbf{r}^{\prime})\bigg)\cdot\nabla w_{h}d\Omega(\mathbf{r}^{\prime})\label{deq:2018d6}
	\end{align}
\end{corollary}

We emphasize that (\ref{deq:2018d6}) is the most fundamental result of the filtering process. It will be used to generate integral representations of the pressure gradient and divergence of the stress tensor.

\subsection{Filtered momentum equation}
Similar to the continuity equation above, we start with the weak form 
\begin{align}
\int_{\Omega}\left(\rho(\mathbf{r}^{\prime})\frac{d\mathbf{u}(\mathbf{r}^{\prime})}{dt}-\nabla^{\prime}\cdot\underline{\underline{\tau}}(\mathbf{r}^{\prime})-\rho(\mathbf{r}^{\prime})\mathbf{b}(\mathbf{r}^{\prime})\right)w_{h}d^{d}\mathbf{r}^{\prime}&=0\quad{} ^\forall w_{h}\in C_{c}^{\infty}(\Omega_{h})\label{deq:2018d7}
\end{align} 
which may be further re-arranged as follows
\begin{align}
\int_{\Omega}\rho(\mathbf{r}^{\prime})\frac{d}{dt}\bigg(\mathbf{u}(\mathbf{r}^{\prime})w_{h}\bigg)d^{d}\mathbf{r}^{\prime}&=\int_{\Omega}\rho(\mathbf{r}^{\prime})\mathbf{u}(\mathbf{r}^{\prime})\frac{dw_{h}}{dt}d^{d}\mathbf{r}^{\prime}\nonumber\\
&\qquad+\langle\nabla\cdot\underline{\underline{\tau}},w_{h}\rangle+\langle\rho_{h}(\mathbf{r})\rangle\widetilde{\mathbf{b}}_{h}(\mathbf{r})\label{deq:2018d8}
\end{align}
where the FIT \ref{prop:conv} has directly been applied to the last term on the right hand side of (\ref{deq:2018d8}). Note the inner product notation for the stress term; this will be expanded shortly. The next thing is to apply the Reynolds transport theorem \ref{theo:rtt} along wih the FIT to the left hand side yielding 
\begin{align}
\frac{d}{dt}\int_{\Omega}\rho(\mathbf{r}^{\prime})\mathbf{u}(\mathbf{r}^{\prime})w_{h}d^{d}\mathbf{r}^{\prime}&=\int_{\Omega}\rho(\mathbf{r}^{\prime})\mathbf{u}(\mathbf{r}^{\prime})\frac{dw_{h}}{dt}d^{d}\mathbf{r}^{\prime}\nonumber\\
&\qquad+\langle\nabla\cdot\underline{\underline{\tau}},w_{h}\rangle+\langle\rho_{h}(\mathbf{r})\rangle\widetilde{\mathbf{b}}_{h}(\mathbf{r})\nonumber\\
\frac{d}{dt}\bigg(\langle\rho_{h}(\mathbf{r})\rangle\widetilde{\mathbf{u}}_{h}(\mathbf{r})\bigg)&=\int_{\Omega}\rho(\mathbf{r}^{\prime})\mathbf{u}(\mathbf{r}^{\prime})\frac{dw_{h}}{dt}d^{d}\mathbf{r}^{\prime}\nonumber\\
&\qquad+\langle\nabla\cdot\underline{\underline{\tau}},w_{h}\rangle+\langle\rho_{h}(\mathbf{r})\rangle\widetilde{\mathbf{b}}_{h}(\mathbf{r})\nonumber\\
\langle\rho_{h}(\mathbf{r})\rangle\frac{d}{dt}\widetilde{\mathbf{u}}_{h}(\mathbf{r})+\widetilde{\mathbf{u}}_{h}(\mathbf{r})\rangle\frac{d}{dt}\langle\rho_{h}(\mathbf{r})\rangle&=\int_{\Omega}\rho(\mathbf{r}^{\prime})\mathbf{u}(\mathbf{r}^{\prime})\frac{dw_{h}}{dt}d^{d}\mathbf{r}^{\prime}\nonumber\\
&+\langle\nabla\cdot\underline{\underline{\tau}},w_{h}\rangle+\langle\rho_{h}(\mathbf{r})\rangle\widetilde{\mathbf{b}}_{h}(\mathbf{r})\label{deq:2018d9}
\end{align}
By substituting (\ref{deq:2018d4}) and applying the chain rule of calculus to the right hand side, we obtain the general form of the filtered momentum equation.
\begin{align}
\langle\rho_{h}(\mathbf{r})\rangle\frac{d}{dt}\widetilde{\mathbf{u}}_{h}(\mathbf{r})&=-\int_{\Omega}\rho(\mathbf{r}^{\prime})\bigg(\widetilde{\mathbf{u}}_{h}(\mathbf{r})-\mathbf{u}(\mathbf{r}^{\prime})\bigg)\otimes\bigg(\widetilde{\mathbf{u}}_{h}(\mathbf{r})-\mathbf{u}(\mathbf{r}^{\prime})\bigg)\cdot\nabla w_{h}d^{d}\mathbf{r}^{\prime}\nonumber\\
&{}\qquad+\langle\nabla\cdot\underline{\underline{\tau}},w_{h}\rangle+\langle\rho_{h}(\mathbf{r})\rangle\widetilde{\mathbf{b}}_{h}(\mathbf{r})\label{deq:2018d10}
\end{align}
Again, the filtered momentum equation is consistent with explicit LES. 
\begin{mydef}[sub-particle stress tensor, SPS]\label{def:sgs}
	The filtered momentum equation (\ref{deq:2018d10}) introduces momentum transfer due to small scale motion. This is defined by the following 
	\begin{align}
	\langle\underline{\underline{\mathcal{H}}}_{h}(\mathbf{r})\rangle&=\int_{\Omega(\mathbf{r})}\rho(\mathbf{r}^{\prime})(\mathbf{u}(\mathbf{r}^{\prime})-\widetilde{\mathbf{u}}_{h}(\mathbf{r}))\otimes(\mathbf{u}(\mathbf{r}^{\prime})-\widetilde{\mathbf{u}}_{h}(\mathbf{r}))w_{h}d\Omega(\mathbf{r}^{\prime})\nonumber\\&=\langle\rho_{h}(\mathbf{r})\rangle\left(\widetilde{(\mathbf{u}\otimes\mathbf{u})}_{h}(\mathbf{r})-\widetilde{\mathbf{u}}_{h}(\mathbf{r})\otimes\widetilde{\mathbf{u}}_{h}(\mathbf{r})\right)\label{deq:2018d11}\quad\text{by the FIT}
	\end{align}
\end{mydef}
which is a symmetric tensor, a property that is fundamental to the global conservation of energy. We note that the presence of the thermo-kinetic stresses in the momentum equation is purely due to the localization of the flow field by the filtering process. In Large Eddy Simulation, LES it is also called the sub-grid stress (SGS) tensor. it is related to the correlation of the velocity components at scales smaller than the filtering dimension $h$. Since the filtering operation filters out high frequency components of the the flow field $\{\rho, p, \mathbf{u}\}$, the SGS tensors captures the effect of small scale motion on the mean flow $\{\langle\rho_{h}\rangle, \langle p_{h}\rangle, \widetilde{\mathbf{u}}_{h}\}$.       

As a consequence of the FIT, the SGS tensor (\ref{deq:2018d11}) has the following interesting property with respect to the action of the divergence operation on it;  
\begin{align}
\nabla\cdot\langle\underline{\underline{\mathcal{H}}}_{h}(\mathbf{r})\rangle&=\int_{\Omega(\mathbf{r})}\rho(\mathbf{r}^{\prime})(\mathbf{u}(\mathbf{r}^{\prime})-\widetilde{\mathbf{u}}_{h}(\mathbf{r}))\otimes(\mathbf{u}(\mathbf{r}^{\prime})-\widetilde{\mathbf{u}}_{h}(\mathbf{r}))\cdot\nabla w_{h}d\Omega(\mathbf{r}^{\prime})\label{deq:2018d11v}
\end{align} 
\begin{proof}
	\begin{align*}
	&\int_{\Omega}\rho(\mathbf{r}^{\prime})(\mathbf{u}(\mathbf{r}^{\prime})-\widetilde{\mathbf{u}}_{h}(\mathbf{r}))\otimes(\mathbf{u}(\mathbf{r}^{\prime})-\widetilde{\mathbf{u}}_{h}(\mathbf{r}))\cdot\nabla w_{h}d^{d}\mathbf{r}^{\prime}\\&
	=\nabla\cdot\int_{\Omega}\rho(\mathbf{r}^{\prime})\mathbf{u}(\mathbf{r}^{\prime})\otimes\mathbf{u}(\mathbf{r}^{\prime})w_{h}d^{d}\mathbf{r}^{\prime}-\bigg(\nabla\cdot\int_{\Omega}\rho(\mathbf{r}^{\prime})\mathbf{u}(\mathbf{r}^{\prime})w_{h}d^{d}\mathbf{r}^{\prime}\bigg)\otimes\widetilde{\mathbf{u}}_{h}(\mathbf{r})\\&
	\qquad -\widetilde{\mathbf{u}}_{h}(\mathbf{r})\otimes\nabla\cdot\int_{\Omega}\rho(\mathbf{r}^{\prime})\mathbf{u}(\mathbf{r}^{\prime})w_{h}d^{d}\mathbf{r}^{\prime}+\widetilde{\mathbf{u}}_{h}(\mathbf{r})\otimes\widetilde{\mathbf{u}}_{h}(\mathbf{r})\cdot\nabla\int_{\Omega}\rho(\mathbf{r}^{\prime})w_{h}d^{d}\mathbf{r}^{\prime}\\
	&= \nabla\cdot\bigg(\langle\rho_{h}\rangle\widetilde{\left(\mathbf{u}\otimes\mathbf{u}\right)}_{h}\bigg)-\bigg(\nabla\cdot\left(\langle\rho_{h}\rangle\widetilde{\mathbf{u}}_{h}\right)\bigg)\otimes\widetilde{\mathbf{u}}_{h}-\widetilde{\mathbf{u}}_{h}\otimes\nabla\cdot\bigg(\langle\rho_{h}\rangle\widetilde{\mathbf{u}}_{h}\bigg)\nonumber\\
	&\qquad+\widetilde{\mathbf{u}}_{h}\otimes\widetilde{\mathbf{u}}_{h}\cdot\nabla\langle\rho_{h}\rangle\\
	&= \nabla\cdot\bigg[\langle\rho_{h}\rangle\bigg(\widetilde{\left(\mathbf{u}\otimes\mathbf{u}\right)}_{h}-\widetilde{\mathbf{u}}_{h}\otimes\widetilde{\mathbf{u}}_{h}\bigg)\bigg]\\
	&=
	\nabla\cdot\langle\underline{\underline{\mathcal{H}}}_{h}(\mathbf{r})\rangle\quad\text{by (\ref{deq:2018d11}) above}
	\end{align*}
	as claimed.
\end{proof}
The SGS tensors are the components of a second order tensor, which is obviously symmetric. The diagonal components are normal stresses whereas the off-diagonal components are shear stresses. The density weighted turbulent kinetic energy $\widetilde{k}_{h}$ is defined to be half the trace of the SGS tensor and is thus
\begin{align}
\langle\rho_{h}(\mathbf{r})\rangle\widetilde{k}_{h}(\mathbf{r})&:=\frac{1}{2}\tr\bigg(\langle\underline{\underline{\mathcal{H}}}_{h}(\mathbf{r})\rangle\bigg)\nonumber\\
&=\frac{1}{2}\int_{\Omega(\mathbf{r})}\rho(\mathbf{r}^{\prime})\vert\vert\widetilde{\mathbf{u}}_{h}(\mathbf{r})-\mathbf{u}(\mathbf{r}^{\prime})\vert\vert^{2}w_{h}d^{d}\mathbf{r}^{\prime}\label{deq:2018d11a}
\end{align}
It is the locally averaged kinetic energy per unit mass of the fluctuating velocity field.

By direct filtering of the energy conservation law (\ref{eq2016:2aa}), a sub-grid term called turbulent dissipation rate appears in the filtered energy balance equation. It is the rate at which turbulent kinetic energy is converted to thermal internal energy of the system. This is defined as
\begin{align}
\langle\rho_{h}(\mathbf{r})\rangle\widetilde{\varepsilon}_{h}(\mathbf{r})&=\int_{\Omega}\rho(\mathbf{r}^{\prime})\nu\underline{\underline{\hat{\sigma}}}_{h}(\mathbf{r}^{\prime}):\nabla^{\prime}\hat{\mathbf{u}}_{h}(\mathbf{r}^{\prime}) d^{d}\mathbf{r}^{\prime}\label{deq:2018m1}
\end{align}
where $\nu^{\textup{eff}}$ is the effective kinematic viscosity and the small-scale dissipation function is given as
\begin{align}
\nu\underline{\underline{\hat{\sigma}}}_{h}(\mathbf{r}):\nabla\hat{\mathbf{u}}_{h}(\mathbf{r})&=-\int_{\Omega}\nu^{\textup{eff}}\bigg(\frac{\langle\rho_{h}(\mathbf{r}^{\prime})\rangle}{\rho(\mathbf{r})}+\frac{\langle\rho_{h}(\mathbf{r})\rangle}{\rho(\mathbf{r}^{\prime})}\bigg)\vert\vert\widetilde{\mathbf{u}}_{h}(\mathbf{r})-\mathbf{u}(\mathbf{r}^{\prime})\vert\vert^{2}\nonumber\\
&\qquad\times\frac{(\mathbf{r}-\mathbf{r}^{\prime})\cdot\nabla\varphi_{h}}{\vert\vert\mathbf{r}-\mathbf{r}^{\prime}\vert\vert^{2}}d^{d}\mathbf{r}^{\prime}\label{deq:2018m2}
\end{align}

It is important to note that the turbulent dissipation function  satisfies the physical requirement that it be negative definite i.e. $\widetilde{\varepsilon}_{h}\leq0$. 

The integro-differential forms for the pressure (\ref{eq2016:2aa}) can be derived in a similar way. With brevity this derivation will be omitted. The complete filtered differential and integro-differential forms of the CNSEs can now be expressed as
\begin{eqnarray}
\frac{d}{dt}\langle\rho_{h}(\mathbf{r})\rangle&=&-\langle\rho_{h}(\mathbf{r})\rangle\nabla\cdot\widetilde{\mathbf{u}}_{h}(\mathbf{r})\label{eqn:2017c2}\\
\frac{d }{dt}\langle p_{h}(\mathbf{r})\rangle&=&-\langle K_{S}\nabla\cdot\mathbf{u},w_{h}\rangle + \gamma\alpha \langle K_{S}\nabla\cdot(\kappa_{s}\nabla p),w_{h}\rangle\label{eqn:2017c3}\\
\langle\rho_{h}\rangle\frac{d}{dt}\widetilde{\mathbf{u}}_{h}&=&\langle\nabla\cdot \underline{\underline{\tau}},w_{h}\rangle-\nabla\cdot\langle \underline{\underline{\mathcal{H}}}_{h}\rangle+ \langle\rho_{h}\rangle\widetilde{\mathbf{b}}_{h}\label{eqn:2017c4}\\
\frac{d\mathbf{r}}{dt} &=& \widetilde{\mathbf{u}}_{h}(\mathbf{r})\label{eqn:2017c5}
\end{eqnarray}

The filtered equations (\ref{eqn:2017c2}), (\ref{eqn:2017c3}), (\ref{eqn:2017c4}), (\ref{eqn:2017c5}) are not closed because the SGS tensor contains unfiltered velocities. To overcome this closure problem several models including  the eddy viscosity model, Smagorinsky model and Germano dynamic model are widely used in LES. 

By proper choice of the FIT, we obtain an SPH model consistent with explicit LES where the smoothed differential operators in (\ref{eqn:2017c2}), (\ref{eqn:2017c3}), (\ref{eqn:2017c4}) can now be expressed in integral form as;
\begin{align}
\langle\rho_{h}(\mathbf{r})\rangle\nabla\cdot\widetilde{\mathbf{u}}_{h}(\mathbf{r})&=-\int_{\Omega}\rho(\mathbf{r}^{\prime})\left(\widetilde{\mathbf{u}}(\mathbf{r})-\mathbf{u}(\mathbf{r}^{\prime})\right)\cdot\nabla w_{h}d^{d}\mathbf{r}^{\prime} \label{eqn:2017h1}\\
\nabla\cdot\langle\underline{\underline{\mathcal{H}}}_{h}(\mathbf{r})\rangle&=\int_{\Omega}\rho(\mathbf{r}^{\prime})(\mathbf{u}(\mathbf{r}^{\prime})-\widetilde{\mathbf{u}}_{h}(\mathbf{r}))\otimes(\mathbf{u}(\mathbf{r}^{\prime})-\widetilde{\mathbf{u}}_{h}(\mathbf{r}))\cdot\nabla w_{h}d^{d}\mathbf{r}^{\prime}\label{eqn:2017h2}\\
\langle\nabla\cdot \underline{\underline{\tau}},w_{h}\rangle&= \int_{\Omega}\underline{\underline{\tau}}(\mathbf{r}^{\prime})\cdot\nabla w_{h}d^{d}\mathbf{r}^{\prime}+\text{Surface terms or S.T.}\label{eqn:2017h3}\\
\nabla\cdot(\kappa_{s}\nabla \langle p_{h}\rangle)&=-2\int_{\Omega}(\kappa_{s}(\mathbf{r})+\kappa_{s}(\mathbf{r}^{\prime}))[\langle p_{h}(\mathbf{r})\rangle-p(\mathbf{r}^{\prime})]\frac{(\mathbf{r}-\mathbf{r}^{\prime})\cdot\nabla w_{h}}{\vert\vert\mathbf{r}-\mathbf{r}^{\prime}\vert\vert^{2}}d^{d}\mathbf{r}^{\prime}\label{eqn:2017h4}\nonumber\\
&\qquad + \text{S.T.}
\end{align}

Using these integral operators, one can then invoke particle discretization using the SPH quadrature: integrals are replaced by summations. Unfortunately, applying the FIT to the CNSEs introduces additional variables: the number of variables is now doubled and therefore the system is no longer closed. The most natural approach for closure is via approximate deconvolution methods (ADM) to be presented elsewhere.

Since this model does not conserve energy, various techniques can however be used to replace the space operators in (\ref{eqn:2017h3}) by momentum conserving operators.

\section{Concept of Fluctuations and Filtering Regularization}
Consider the fluid domain $\Omega$ to be unbounded or that $supp(w^{h})\cap\partial\Omega=\emptyset$, then the FIT of proposition \ref{prop:conv} can be expressed as follows.  
\begin{myprop}[FIT on unbounded domains]\label{prop:fconv}
	Let $\Omega_{h}(\mathbf{r})$ be a locally compact space within the fluid domain $\Omega$. Then the filtered mass density, momentum density and pressure are given by the FIT; for each $w_{h}\in C^{\infty}_{0}(\Omega_{h})$
	\begin{align}
	\langle\rho_{h}(\mathbf{r})\rangle &=\rho(\mathbf{r})-\int_{\Omega}\bigg(\rho(\mathbf{r})-\rho(\mathbf{r}^{\prime})\bigg)w_{h}(\mathbf{r}-\mathbf{r}^{\prime})d^{d}\mathbf{r}^{\prime}\label{ddeq:1}\equiv \rho(\mathbf{r})-\hat{\rho}_{h}(\mathbf{r})\\
	\widetilde{\mathbf{u}}_{h}(\mathbf{r}) &=\mathbf{u}(\mathbf{r})-\frac{1}{\langle\rho_{h}(\mathbf{r})\rangle}\int_{\Omega}\rho(\mathbf{r}^{\prime})\bigg(\mathbf{u}(\mathbf{r})-\mathbf{u}(\mathbf{r}^{\prime})\bigg)w_{h}(\mathbf{r}-\mathbf{r}^{\prime})d^{d}\mathbf{r}^{\prime}\label{ddeq:1p}\equiv \mathbf{u}(\mathbf{r})-\hat{\mathbf{u}}_{h}(\mathbf{r})\\
	\langle p_{h}(\mathbf{r})\rangle &=p(\mathbf{r})-\int_{\Omega}\bigg(p(\mathbf{r})-p(\mathbf{r}^{\prime})\bigg)w_{h}(\mathbf{r}-\mathbf{r}^{\prime})d^{d}\mathbf{r}^{\prime}\equiv p(\mathbf{r})-\hat{p}_{h}(\mathbf{r})\label{ddeq:p1}
	\end{align} 
\end{myprop}

Therefore, the FIT is an integral transform that filters out the local fluctuation field \{$\hat{\rho}_{h}, \hat{\mathbf{u}}_{h},\hat{p}_{h}$\} from the underlying disordered field $\{\rho, \mathbf{u}, p\}$ to generate a smooth field $\{\langle{\rho}_{h}\rangle, \widetilde{\mathbf{u}}_{h}, \langle{p}_{h}\rangle$\} field. The smoothing process prevents the production of small scale flow structures due to the fluctuations $\{\hat{\rho}_{h}, \hat{\mathbf{u}}_{h}, \hat{p}_{h}$\}. Under the assumptions that the fluid particles are uniformly distributed locally, the fluctuations may be taken to be small. However, one has to be cautious as this is a very crude approximation. This special case is called the SPH golden rule. To prove this claim, we first prepare the following theorem.

\begin{mydef}[Cauchy-Schwartz inequality]
	Let $\mathbf{u}$, $\mathbf{v}$ be vectors in a vector space $\Omega$ with an inner product. The Cauchy-Schwartz inequality states that
	\begin{align}
	\vert\mathbf{u}\cdot\mathbf{v}\vert\vert\leq\vert\vert\mathbf{u}\vert\vert\vert\vert\mathbf{v}\vert\vert
	\end{align}
	\begin{proof}
		Assuming that $\mathbf{v}\neq 0$, let $\lambda\in\mathbb{C}$ be given by $\lambda:=\mathbf{u}\cdot\mathbf{v}/\vert\vert\mathbf{v}\vert\vert^{2}$, then
		\begin{align}
		0&\leq\vert\vert\mathbf{u}+\lambda\mathbf{v}\vert\vert^{2}\nonumber\\
		&=\vert\vert\mathbf{u}\vert\vert^{2}-\lambda\overline{\mathbf{u}\cdot\mathbf{v}}-\overline{\lambda}(\mathbf{u}\cdot\mathbf{v})+\overline{\lambda}{\lambda}\vert\vert\mathbf{v}\vert\vert^{2}\nonumber\\
		&=\vert\vert\mathbf{u}\vert\vert^{2}-\frac{\vert\mathbf{u}\cdot\mathbf{v}\vert^{2}}{\vert\vert\mathbf{v}\vert\vert^{2}}\nonumber
		\end{align}
		Therefore, $\vert\mathbf{u}\cdot\mathbf{v}\vert\leq\vert\vert\mathbf{u}\vert\vert\vert\vert\mathbf{v}\vert\vert$.  
	\end{proof} 
\end{mydef}
\begin{mydef}[Lipschitz continuity]
	Let $\Omega$ be the fluid domain and $\Omega_{h}(\mathbf{r})=\{\mathbf{r},\mathbf{r}^{\prime}\in\Omega\subset\mathbb{R}^{n}:\vert\vert\mathbf{r}-\mathbf{r}^{\prime}\vert\vert\leq h,~w_{h}\geq0\}$ the test space centered about $\mathbf{r}$. A mapping $\Psi :\mathbb{R}^{d}\to\mathbb{R}^{n}$ is Lipschitz continuous of $\Omega_{h}(\mathbf{r})$ if $^{\exists}M\ge 0$ such that $^{\forall}\mathbf{r},\mathbf{r}^{\prime}\in\Omega_{h}(\mathbf{r})\subset\Omega$
	\begin{align}
	\vert \rho(\mathbf{r})-\rho(\mathbf{r}^{\prime})\vert\le M\vert\vert\mathbf{r}-\mathbf{r}^{\prime}\vert\vert
	\end{align}
	\begin{proof}
		Define $\varphi(s):[0,1]\mapsto\Omega_{h}(\mathbf{r})$ in the following way:
		\begin{align*}
		\varphi(s):=\mathbf{r}^{\prime}+(\mathbf{r}-\mathbf{r}^{\prime})s
		\end{align*}
		This function is differentiable on $(0,1)$ and continuous on $[0,1]$ and so is the composite function $\rho\circ\varphi$. According to the mean value theorem for a single variable function, there exists a $\xi\in [0,1]$ such that 
		\begin{align*}
		(\rho\circ\varphi)^{\prime}(\xi) = (\rho\circ\varphi)(1)-(\rho\circ\varphi)(0) 
		\end{align*}
		By the chain rule of composite functions, it then follows that
		\begin{align*}
		\rho(\mathbf{r})-\rho(\mathbf{r}^{\prime})&=\nabla\rho(\varphi({s}))\cdot(\mathbf{r}-\mathbf{r}^{\prime})\\
		\vert\rho(\mathbf{r})-\rho(\mathbf{r}^{\prime})\vert&\leq\vert\vert\nabla\rho(\varphi({s}))\vert\vert\vert\vert\mathbf{r}-\mathbf{r}^{\prime}\vert\vert
		\end{align*}
		Then ${\forall}\mathbf{r},\mathbf{r}^{\prime}\in\Omega_{h}(\mathbf{r})\subset\Omega$, since the gradient of the fluid density is bounded, it follows that $^{\exists}M\ge0$ so that $\vert\vert\nabla\rho(\varphi({s}))\vert\vert\leq M$. We say that $\rho$ is Lipschitz continuous with
		\begin{align*}
		\vert\rho(\mathbf{r})-\rho(\mathbf{r}^{\prime})\vert&\leq\vert\vert\nabla\rho(\varphi({s}))\vert\vert\vert\vert\mathbf{r}-\mathbf{r}^{\prime}\vert\vert\leq M\vert\vert\mathbf{r}-\mathbf{r}^{\prime}\vert\vert
		\end{align*}
	\end{proof}
\end{mydef}

\begin{mytheo}[uniform continuity]\label{theo:ucontinuity}
	If the fluid density $\rho:\Omega_{h}(\mathbf{r})\subset\Omega\to\mathbb{R}$ satisfies the Lipschitz inequality $^{\forall}\mathbf{r},\mathbf{r}^{\prime}\in\Omega_{h}(\mathbf{r})\subset\Omega$, then $\rho$ is uniformly continuous.
\end{mytheo}
\begin{proof}
	$^{\exists}\varepsilon > 0$ so that for $ h = \varepsilon/M$, $^{\forall}\mathbf{r},\mathbf{r}^{\prime}\in\Omega_{h}(\mathbf{r})\subset\Omega$, $\vert\vert\mathbf{r}-\mathbf{r}^{\prime}\vert\vert<h$ implies that
	\begin{align*}
	\vert\rho(\mathbf{r})-\rho(\mathbf{r}^{\prime})\vert&\leq M\vert\vert\mathbf{r}-\mathbf{r}^{\prime}\vert\vert=Mh=\varepsilon
	\end{align*}
\end{proof}

\begin{mytheo}[Filtering regularization]\label{theo:fregular}
	The filtering of the fluid density $\rho:\Omega_{h}(\mathbf{r})\subset\Omega\to\mathbb{R}$ by the convolution filter $w_{h}\in\textup{C}_{c}^{\infty}(\Omega)$ yields the local density approximation $\langle\rho_{h}\rangle:\Omega_{h}(\mathbf{r})\subset\Omega\to\mathbb{R}$ as given by the FIT. Then for all $\mathbf{r}\in\Omega_{h}(\mathbf{r})$, in the continuum limit
	\begin{align}
	\lim_{h\to 0}\langle\rho_{h}(\mathbf{r})\rangle &= \rho(\mathbf{r})\Longleftrightarrow \lim_{h\to 0}\hat{\rho}_{h}(\mathbf{r})\rangle = 0
	\end{align}
	uniformly.
	\begin{proof}
		Take $h$ so that $\Omega_{h}(\mathbf{r})$:=supp$(w_{h})=\{\mathbf{r},\mathbf{r}^{\prime}\in\mathbb{R}^{n}:\vert\vert\mathbf{r}-\mathbf{r}^{\prime}\vert\vert\le h,~w_{h}\geq0\}$. We have
		\begin{align*}
		\hat{\rho}_{h}(\mathbf{r})&\stackrel{\textup{def}}{=}\int_{\Omega_{h}(\mathbf{r})} \bigg(\rho(\mathbf{r})-\rho(\mathbf{r}^{\prime})\bigg)w_{h}(\mathbf{r}-\mathbf{r}^{\prime})d^{d}\mathbf{r}^{\prime}\\
		\vert\hat{\rho}_{h}(\mathbf{r})\vert&=\bigg\vert\int_{\Omega_{h}(\mathbf{r})}\rho(\mathbf{r})-\rho(\mathbf{r}^{\prime}) w_{h}(\mathbf{r}-\mathbf{r}^{\prime})d^{d}\mathbf{r}^{\prime}\bigg\vert\\
		&\le\int_{\Omega_{h}(\mathbf{r})}\vert\rho(\mathbf{r})-\rho(\mathbf{r}^{\prime})\vert w_{h}(\mathbf{r}-\mathbf{r}^{\prime})d^{d}\mathbf{r}^{\prime}\\
		&\le\text{uniform continuity} \\
		&\le \varepsilon\int_{\Omega_{h}(\mathbf{r})}w_{h}(\mathbf{r}-\mathbf{r}^{\prime})d^{d}\mathbf{r}^{\prime}\\
		\therefore\vert\hat{\rho}_{h}(\mathbf{r})\vert&\leq\varepsilon
		\end{align*}
	\end{proof} 
\end{mytheo}
By the filtering regularization theorem \ref{theo:fregular} above, it follows that if the local distribution of fluid particles is uniformly continuous then the local density fluctuations are arbitrarily small (vanishing in the continuum limit). 

By the filtering regularization theorem \ref{theo:fregular} above, provided that $h$ is small then the following zeroth order deconvolution approximation is satisfactory.  
\begin{align}
\langle\rho_{h}\rangle\approx\rho,\quad \langle p_{h}\rangle\approx p,\quad \widetilde{\mathbf{u}}_{h}\approx\mathbf{u}\label{deq:2018q1}
\end{align}
with acquires order of accuracy $\mathcal{O}(h^2)$ for bell-shaped convolution filters assuming that the local particle distribution on $\Omega_{h}(\mathbf{r})$:=supp$(w_{h})$ is uniformly continuous and $\Omega_{h}$ does not intersect the domain boundaries. Consequently, after particle discretization, the convergence obtained might not remain as favorable. These approximations are very crude \cite{Admal2010,Quinlan2006}. Furthermore, this inaccuracy will lead to difficulties in enforcing essential boundary conditions.

\section{Choice of velocity smoothing}
In recent years, Monaghan \cite{Monaghan2002}, \cite{Monaghan2011} has pioneered the development of a turbulent model for smoothed particle hydrodynamics called SPH$-\epsilon$. Similar to Monaghan's model, Hu et.al. \cite{Hu2015}have recently proposed a model called SPH$-\sigma$. Due to the adopted velocity smoothing approach in both models, they are fundamentally similar to the Lagrangian Averaged Navier-Stokes, LANS-$\alpha$ model. 

The choice of velocity smoothing is very important especially in the study of turbulence as it leads to varying mathematical structure of the filtered equations. Here a comparison of the velocity smoothing of Monaghan's SPH$-\epsilon$ with the FIT proposed in this thesis are compared. 
\begin{align}
\hat{\mathbf{u}}(\mathbf{r}) = \mathbf{u}(\mathbf{r})-\epsilon\int \bigg(\mathbf{u}(\mathbf{r})-\mathbf{u}(\mathbf{r}^{\prime})\bigg)w_{h}(\mathbf{r}-\mathbf{r}^{\prime})d\mathbf{r}^{\prime}\label{eqn:201612b}
\end{align}
The parameter $\epsilon$ is a constant $0\leq\epsilon\leq 1$. It can be shown that this smoothing approach is similar to the one we propose in this paper. Consider the FIT for unbounded domains \ref{prop:fconv}, 
\begin{align}
\widetilde{\mathbf{u}}_{h} = \mathbf{u}(\mathbf{r})-\frac{1}{\langle\rho_{h}(\mathbf{r})\rangle}\int\rho(\mathbf{r}^{\prime}) \bigg(\mathbf{u}(\mathbf{r})-\mathbf{u}(\mathbf{r}^{\prime})\bigg)w_{h}(\mathbf{r}-\mathbf{r}^{\prime})d\mathbf{r}^{\prime}\label{eqn:201612a}
\end{align} 
which leads to the explicit Large Eddy Simulation (LES) model for compressible flows.   

As a consequence of the choice of velocity smoothing, we inspect Monaghan's SPH-$\epsilon$ model in which their choice of velocity smoothing (\ref{eqn:201612b}) as well as derivation from a Lagrangian leads to a scalar turbulence term in the momentum given as
\begin{align}
\rho_{0}\frac{d}{dt}\mathbf{u}(\mathbf{r})=-\frac{\epsilon}{2}\int\rho(\mathbf{r}^{\prime})\vert\mathbf{u}(\mathbf{r})-\mathbf{u}(\mathbf{r}^{\prime})\vert^{2}\nabla w_{h}d\mathbf{r}^{\prime}\label{eqn:201612c}
\end{align} 
which is the acceleration due to turbulence. This term only modifies the pressure and does not contribute to the shear viscosity. On the other hand the thermokinetic stress quantity in our model is a tensorial quantity given by definition (\ref{def:sgs}). In particular, the turbulence acceleration in this case is 
\begin{align}
&\langle\rho_{h}(\mathbf{r})\rangle\frac{d}{dt}\widetilde{\mathbf{u}}_{h}(\mathbf{r})\nonumber\\
& \qquad {} = -\int_{\Omega_{h}(\mathbf{r})}\rho(\mathbf{r}^{\prime})(\mathbf{u}(\mathbf{r})-\widetilde{\mathbf{u}}_{h}(\mathbf{r}^{\prime}))\otimes(\mathbf{u}(\mathbf{r})-\widetilde{\mathbf{u}}_{h}(\mathbf{r}^{\prime}))\cdot\nabla \varphi_{h}d\Omega(\mathbf{r}^{\prime})\label{eqn:201612d}
\end{align}
Therefore the turbulence term in our model takes into account flow directionality as it is a tensor. On the other hand, (\ref{eqn:201612c}) is a scalar (possibly because it was derived from a Lagrangian) and thus cannot take flow directionality into account. Moreover, it only modifies the pressure term in the momentum equation without any contribution to the shear viscosity that is known to dominate the molecular viscosity at high Reynolds number. It can then be argued that Monaghan's model, in its present form, cannot fully account for turbulent phenomena. It is, however, clear that if the off diagonal terms in (\ref{eqn:201612d}) are zero, the kernel is normalized as it should and the fluid is incompressible then the two models are again similar. 

Other forms of velocity smoothing techniques are possible. Unlike the FIT procedure which appears unambiguous, the velocity smoothing can be defined in multiple other ways. It may seem more natural to define the local velocity in the same structural form as the local density and pressure such that
\begin{align}
\langle\mathbf{u}_{h}(\mathbf{r})\rangle &=\int_{\Omega_{h}(\mathbf{r})}\mathbf{u}_{h}(\mathbf{r}^{\prime})w_{h}(\mathbf{r}-\mathbf{r}^{\prime})d^{n}\mathbf{r}^{\prime}\label{eqn:2018a1}
\end{align}   
or a re-normalized form
\begin{align}
\langle\mathbf{u}_{h}(\mathbf{r})\rangle &=\frac{\int_{\Omega_{h}(\mathbf{r})}\mathbf{u}_{h}(\mathbf{r}^{\prime})w_{h}(\mathbf{r}-\mathbf{r}^{\prime})d^{n}\mathbf{r}^{\prime}}{\int_{\Omega_{h}(\mathbf{r})}w_{h}(\mathbf{r}-\mathbf{r}^{\prime})d^{n}\mathbf{r}^{\prime}}\label{eqn:2018a2}
\end{align}
While both definitions (\ref{eqn:2018a1}) and (\ref{eqn:2018a1}) are mathematically valid, they lack physical foundation in principle. Moreover, they both introduce mass transfer due turbulence in the continuity equation leading to a closure problem.

\section{Conclusions}
The FIT has been successfully applied to the CNSEs to derive a set of filtered equations that are consistent with explicit LES. Most importantly, the origin of the SGS tensor has be clearly demonstrated. With a consistent set of filtered equations, closure models used in LES schemes can be adopted for this version of SPH. 

The correct way of moving the target particle under the influence of support particles has also be given. Furthermore, an exact integral formula for the turbulent kinetic energy has been derived from the SGS tensor whereas the associated turbulence dissipation rate has been proposed, based on intuition and deeper understanding of the SPH concepts, without giving any rigorous treatment.  

Finally, the concept of fluctuations has been introduced. A general conclusion is that if the smoothing length $h$ is sufficiently small and that local particle distribution remains uniform throughout the time evolution of the system, then the fluctuations are arbitrarily small and can thus be neglected. In this case the filtered field and the disordered field are approximately the same, hence no close problem regarding the SGS tensor and other aspects of the proposed SPH model. 

We make the following observations regarding the discussion in this paper;
\begin{enumerate}
\item[(1)] The target SPH particle at position $\mathbf{r}$ moves with a field velocity $\widetilde{\mathbf{u}}_{h}(\mathbf{r})$ which is a locally averaged or smoothed velocity of all other fluid particles within the target particle's support $\Omega_{h}(\mathbf{r})$.  
\item[(2)] both the mass and momentum conservation laws and the pressure equation are not closed. Closure models are required to close up the system.  
\item[(3)] With a proper choice of smoothing rules, the governing equations of LES and SPH are equivalent for compressible turbulence. 
\item[(4)] No surface terms show up in the smoothed continuity equation. This is important since SPH formulations with surface terms in the continuity equation have been proposed.
\item[(5)] Due to the FIT, the associated Lagrangian is non trivial. This makes it difficult to derive averaged equations of motion by variational principles.  
\end{enumerate}

\renewcommand{\refname}{\spacedlowsmallcaps{References}} 

\bibliographystyle{unsrt}

\bibliography{articlez} 


\end{document}